\documentclass[a4paper, 10pt, conference]{article}      

\usepackage{graphics} 
\usepackage{epsfig} 
\usepackage{mathptmx} 
\usepackage{amsmath} 
\usepackage{amssymb}  
\usepackage{algorithm}
\usepackage{algpseudocode}%

\newtheorem{theorem}{Theorem}
\newtheorem{problem}{Problem}

\newtheorem{proposition}{Proposition}
\newtheorem{lemma}{Lemma}
\newtheorem{proof}{Proof}
\newenvironment{eq}{\everymath {\displaystyle \everymath{ }} \equation}{ \endequation} %

\providecommand{\eval}[2]{\left.#1\right\rvert_{#2}}
\newcommand{\trace}{\textnormal{Trace}}

\DeclareMathOperator*{\argmin}{argmin}

\title{\LARGE \bf
Realization independent single time-delay dynamical model interpolation and $\mathcal{H}_2$-optimal approximation}

\author{I. Pontes Duff, C. Poussot-Vassal and C. Seren
\thanks{I. Pontes Duff is with ISAE, Onera - The French Aerospace Lab, F-31055 Toulouse, France; email:
        {\tt\small ipontes@onera.fr}}%
\thanks{C. Poussot-Vassal and C. Seren are Onera - The French Aerospace Lab, F-31055 Toulouse, France; emails:
        {\tt\small charles.poussot-vassal@onera.fr} and {\tt\small cedric.seren@onera.fr}} %
     }

\begin{document}

\maketitle
\thispagestyle{empty}
\pagestyle{empty}

\begin{abstract}
In this paper, the realization-free model approximation problem, as stated in \cite{mayo2007framework,beattie2012realization}, is revisited in the case where the interpolating model might be time-delay dependent. To this aim, the Loewner framework, initially settled for delay-free realization, is firstly generalized to the  single delay case. Secondly,  the (infinite) model approximation $\mathcal{H}_2$ optimality conditions are established through the use of the Lambert functions. Finally, a numerically effective iterative scheme, named \textbf{dTF-IRKA}, similar to the \textbf{TF-IRKA} \cite{beattie2012realization}, is proposed to reach a part of the aforementioned optimality conditions. The proposed method validity and interest are assessed on different numerical examples. 
\end{abstract}

\section{Introduction }

\subsection{Motivating context and problem formulation}

Mathematical dynamical models are usually necessary to understand, analyse and control the behaviour of physical phenomena. Usually, high fidelity models require numerous equations and variables. The resulting associated state-space realization is consequently of large dimensions and the resulting system is said to be a large-scale one. In addition, in some cases, a linear finite dimension realization is not accessible or even does not exist (\emph{e.g.}, partial differential equations models, irrational transfer functions, etc.). Although such models can faithfully and accurately reproduce reality, it might lead to \emph{(i)} high complexity and/or \emph{(ii)} realization-less models for which classical methods cannot be reasonably applied due to high numerical burden, low computational speed and inappropriate tools. In these cases and for control concerns, an approximation by a less complex realization is therefore desirable. 
This justifies the use of model approximation and interpolation techniques allowing to find a simpler model which faithfully approaches the original one and that can be used in place for simulation, analysis and control (for survey and historical references, see \cite{Luenberger1967, wilson1970optimum, antoulas2005approximation} and references therein).

Moreover, as Time-Delay Systems (TDS) is a large class of dynamical systems which generalizes the finite dimension realization one, \emph{approximating any transfer functions}\footnote{Throughout this paper, we denote $\mathcal{H}_2^{(n_y\times n_u)}$ or simply $\mathcal{H}_2$, the open subspace of $\mathcal{L}_2$ with matrix-valued function $H(s)$ ($n_y$ outputs, $n_u$ inputs), $\forall s \in \mathbb{C}$, which are analytic in $\textbf{Re}(s)> 0$ (functions that are locally given by a convergent power series and differentiable on each point of its definition set) \cite{PontesECC:2015}.} $H(s) \in\mathcal{H}_2^{(n_y\times n_u)}$ or complex data sets\footnote{We denote as $\big(s_i,H(s_i)\big)$ the evaluation of transfer $H$ at $s_i$.} $\big(s_i,H(s_i)\big)$ (for $i=1,\dots,r$, $r\in \mathbb{N}^*$), \emph{by a time-delay dynamical model} might be relevant for some specific applications where the delay naturally appears. Indeed for such kind of systems many dedicated and powerful results have been obtained for stability, performance analysis and control (see \emph{e.g.}, \cite{richard2003time, Niculescu:01,Briat:14e}). 
 
In this paper, the approximation of any realization or realization-free linear dynamical model by a \emph{single time-delay model of finite dimension} is developed. More specifically, we are interested in approximating any MIMO transfer function $H(s) \in \mathcal{H}_2$ by a single delay finite-dimensional linear time-invariant descriptor system denoted  $\mathbf{H_d} = (E,A,B,C,\tau)$ and defined by:
\begin{equation}\label{eq:SSDescriptordelay}
E\dot{x}(t) = Ax(t-\tau) +Bu(t),~y(t) = Cx(t),
\end{equation}    
whose transfer function  is $H_d(s) = C(sE-Ae^{-\tau s})^{-1}B$. It is straightforward to note that the approximation form \eqref{eq:SSDescriptordelay} generalizes the delay-free one used in \cite{mayo2007framework,beattie2012realization} given as $\mathbf{H} = (E,A,B,C,0)$ (or simply $(E,A,B,C)$), 
\begin{equation}\label{eq:SSDescriptor}
E\dot{x}(t) = Ax(t) +Bu(t),~y(t) = Cx(t).
\end{equation} 

Following \cite{mayo2007framework, beattie2012realization}, and inspired by the widely used $\mathcal{H}_2$-approximation problem \cite{benner2005dimension,dooren2007,gugercin2008h_2}, our objective can be mathematically formulated as follows:  
\begin{problem}\label{pb:General} Given a LTI system $H(s) \in \mathcal{H}_2$ (or $\big(s_i,H(s_i)\big)$, the evaluation of $H(s)$ at $s_i\in\mathbb{C}$, for $i=1,\dots,r$), a positive integer $r \in \mathbb{N}^*$ and a positive scalar $\tau \in \mathbb{R}$, find a model $\mathbf{\hat{H}_d} = (E,A,B,C,\tau)\in \mathcal{H}_{2}$  such that
\begin{equation}
\mathbf{\hat{H}_d} := \argmin_{\mathbf{G_d}\in \mathcal{H}_2 ,\dim(\mathbf{G_d} ) \leq r} \|\mathbf{H}-\mathbf{G_d}\|_{\mathcal{H}_2}.
\end{equation}
\end{problem} 

In other words, if an evaluation of the transfer function $H(s)$, for any $s\in\mathbb{C}$, is available (either from data or by simply evaluating $H(s)$), our goal is to find a delay model of the form  \eqref{eq:SSDescriptordelay}, that well approximates $H$, in the sense of the $\mathcal{H}_2$-norm.


\subsection{Contributions}

The purpose of this paper is thus to extend the application domain of the Loewner framework established in \cite{mayo2007framework,IonitaPhd2013} to dynamical systems with one single internal delay. To do so, a new \emph{delay Loewner framework} is firstly developed to interpolate a given transfer function by a time-delay model of the form \eqref{eq:SSDescriptordelay}, enabling the delayed Loewner framework to be applied to any models for which the transfer function is  accessible only. This allows then model approximation for both infinite/finite dimensional systems and data-based ones. Then, following Problem \ref{pb:General}, the \emph{$\mathcal{H}_2$-oriented optimality conditions} are formulated and used to construct an iterative algorithm, similar to the recently proposed \textbf{TF-IRKA} \cite{beattie2012realization}, allowing to obtain an approximated model $\mathbf{\hat{H}_d}$ satisfying \emph{a finite number of the $\mathcal{H}_2$ optimality conditions}. 

\subsection{Notations and outlines}

We denote by $\mathbb{N}^*$ the set of natural numbers without 0, by  $\mathcal{H}_2$ the Hilbert space of matrix-valued functions $F: \mathbb{C} \rightarrow \mathbb{C}^{n_y\times n_u}$ satisfying $\int_{\mathbb{R}}\trace[\overline{F(i\omega)}F(i\omega)^T]d\omega <\infty$ whose components $f_{i,j}$ are analytic in the open right half plane. For $\mathbf{H},\mathbf{G}\in \mathcal{H}_2(i\mathbb{R})$, we define the inner-product
\[\langle\textbf{H} ,\textbf{G} \rangle_{\mathcal{H}_2} = \int_{-\infty}^{\infty}\textnormal{trace}\Big(\overline{H(i\omega)}G(i\omega)^T\Big)d\omega, \]
with corresponding induced-norm $\|\mathbf{H}\|_{\mathcal{H}_2} = \langle \mathbf{H},\mathbf{H}\rangle_{\mathcal{H}_2}^{\frac{1}{2}}$. Let finally denote by $F'(\lambda) = \eval{dF/ds}{s=\lambda}$.

The paper is organized as follows:  Section \ref{sec:reviewinterpolation} recalls some preliminary results on the rational interpolation Loewner framework proposed in \cite{mayo2007framework}.  Section \ref{sec:delayinterpolation}  presents the extension of these results to the  single-delay case. Section \ref{sec:Optimality} derives the first order optimality conditions from the $\mathcal{H}_2$-optimisation Problem \ref{pb:General}. Then, Section \ref{sec:delayTFIRKA} details an iterative algorithm celebrated as \textbf{dTF-IRKA}\footnote{\textbf{dTF-IRKA} stands for delay Transfer Function Iterative Rational Krylov Algorithm.} (inspired by the \textbf{TF-IRKA} from \cite{beattie2012realization}), which allows to obtain an approximation satisfying some optimality conditions in a numerically efficient and memory affordable way.  Finally, Section \ref{sec:applications} illustrates the proposed approach and framework on numerical examples.
\section{Realization-less interpolation}\label{sec:reviewinterpolation}

\subsection{ Preliminary results in Loewner framework for rational interpolation} 
The interpolation problem, in its basic general form, is stated as follows:
\begin{problem}[General interpolation problem \cite{mayo2007framework}] \label{pb:generalinterp} Given \emph{right}:
\begin{equation}\label{eq:rightinterdata}
\{(\lambda_i,\mathbf{r}_i,\mathbf{w}_i) | \lambda_i \in \mathbb{C}, \mathbf{r}_i \in \mathbb{C}^{n_u \times 1} , \mathbf{w}_i \in \mathbb{C}^{n_y \times 1}, i =1,\dots, \rho  \}
\end{equation} 
and \emph{left interpolation data}:
\begin{equation}\label{eq:leftinterdata}
\{(\mu_j,\mathbf{l}_j,\mathbf{v}_j) | \mu_j \in \mathbb{C}, \mathbf{l}_j \in \mathbb{C}^{1 \times n_u} , \mathbf{v}_j \in \mathbb{C}^{1 \times n_y}, j =1,\dots, \nu \}
\end{equation}
construct a realization $\mathbf{H} = (E,A,B,C)$ of appropriate dimensions whose transfer function $H(s) = C(sE-A)^{-1}B$ both satisfies the \emph{right}:
\begin{equation}
H(\lambda_i)\mathbf{r}_i = \mathbf{w}_i,~ i = 1,\dots \rho
\label{eq:constRight}
\end{equation} 
and the \emph{left constraints}:
\begin{equation}
\mathbf{l}_jH(\mu_j) = \mathbf{v}_j,~ j = 1,\dots \nu.
\label{eq:constLeft}
\end{equation}
\end{problem}

The above problem can be solved (to obtain real-valued matrices) thanks to the following theorem, proposed by \cite{mayo2007framework}.

\begin{theorem}[Loewner framework \cite{mayo2007framework}]\label{thm:Loewner}
Given \emph{right} and \emph{left interpolation data} as in \eqref{eq:rightinterdata}-\eqref{eq:leftinterdata}, and assuming that $\rho = \nu = r$ , the realization $\mathbf{H} = (E,A,B,C)$ of order $r$ constructed as
\begin{equation}
E= -\mathbb{L}, A= -\mathbb{L}_{\sigma}, B=V, C=W,
\end{equation}
interpolates the right and left constraints \eqref{eq:constRight}-\eqref{eq:constLeft}, if
\begin{equation}
\begin{array}{rcl}
[\mathbb{L}]_{ij} \hspace{-0.3cm} &=& \dfrac{\mathbf{v}_i\mathbf{r}_j - \mathbf{l}_i\mathbf{w}_j}{\mu_i - \lambda_j} = \dfrac{\mathbf{l}_i\big( H(\lambda_i) - H(\mu_j) \big) \mathbf{r}_j}{\mu_i - \lambda_j} \\
\hspace{-0.9cm}\,[\mathbb{L}_{\sigma}]_{ij}\hspace{-0.3cm} &=& \dfrac{\mu_i\mathbf{v}_i\mathbf{r}_j - \mathbf{l}_i\mathbf{w}_j\lambda_j}{\mu_i - \lambda_j} =  \dfrac{\mu_i  \mathbf{l}_i\big( H(\lambda_i) - H(\mu_j) \big) \mathbf{r}_j  \lambda_j}{\mu_i - \lambda_j}
\end{array}
\end{equation}
known as the Loewner and the shifted Loewner matrices, respectively, and \[W =[\textbf{w}_1, \dots, \textbf{w}_r] ~,~ V^T =  [\textbf{v}_1, \dots,\textbf{v}_r]. \]
\end{theorem}

Theorem \ref{thm:Loewner}  allows to obtain a model $\mathbf{H} = (E,A,B,C)$ whose transfer function interpolates right and left constraints as stated in Problem \ref{pb:generalinterp}. This has been extensively used for system identification from complex data obtained by a signal generator and for large-scale model approximation purposes \cite{IonitaPhd2013},\cite{ionita2014data}. An extension of Problem \ref{pb:generalinterp}, including some derivative constraints, has also been considered to solve the $\mathcal H_2$ model approximation problem \cite{beattie2012realization}. To this aim, the following theorem, initially stated in \cite{mayo2007framework}, provides a solution for the problem with derivatives constraints in the case where the right and left interpolation points are equals, \emph{i.e.}, $s_i = \mu_i = \lambda_i~,~ \forall i = 1,\dots r$. 

\begin{theorem}[Derivative Loewner framework \cite{mayo2007framework}]\label{thm:DerivativeLoewner}
Given a system represented by its transfer function $H(s)$,  $r$ shift points $\{s_1,\dots,s_r\} \in \mathbb{C}$ and $r$ tangential directions $\{\mathbf{l}_1,\dots, \mathbf{l}_r\} \in \mathbb{C}^{1\times n_y}$, $\{\mathbf{r}_1,\dots,\mathbf{r}_r\}\in \mathbb{C}^{n_u \times 1}$, the $r$-dimensional descriptor model $\mathbf{\hat{H}} = (\hat{E},\hat{A},\hat{B},\hat{C})$, as in \eqref{eq:SSDescriptor}, interpolates $H(s)$ as follows, for $k=1,\dots r$:
\begin{equation}
 H({s}_k)\mathbf{r}_k = \hat{H}({s}_k)\mathbf{r}_k,\,\,\mathbf{l}_kH(s_k) = \mathbf{l}_k\hat{H}({s}_k)\,
\label{pbh2}
\end{equation}
\begin{equation}
\mathbf{l}_kH'(s_k)\mathbf{r}_k = \mathbf{l}_k\hat{H}'(s_k)\mathbf{r}_k,
\label{pbh21}
\end{equation}
 if for $i,j=1,\dots,r$:
\begin{eq}
(\hat{E})_{ij} =
\left\{\begin{array}{lr}
-\frac{\mathbf{l}_i \big(H(s_i)-H(s_j)\big)\mathbf{r}_j}{s_i-s_j} & i\neq j\\
-\mathbf{l}_i H'(s_i)\mathbf{r}_i & i=j
\end{array}\right.\nonumber
\end{eq}
\begin{eq}
(\hat{A})_{ij} = \left\{\begin{array}{lr}
-\frac{\mathbf{l}_i \big(s_iH(s_i)-s_jH(s_j)\big)\mathbf{r}_j}{s_i-s_j} & i\neq j\\
-\mathbf{l}_i \big(sH(s)\big)'|_{s=s_i}\mathbf{r}_i& i=j
\end{array}\right. \nonumber
\end{eq}
\begin{equation}
\hat{C} = [H(s_1)\mathbf{r}_1,\dots ,H(s_r)\mathbf{r}_r] \hspace{0.2cm} \textnormal{and} \,\,\hat{B} = \left[\begin{array}{ccc}
\mathbf{l}_1H(s_1) \\
\vdots\\
\mathbf{l}_rH(s_r) \nonumber
\end{array}\right].
\end{equation}
\end{theorem}

In the following section the extensions of both Theorems \ref{thm:Loewner} and \ref{thm:DerivativeLoewner} are presented in the case where a time-delay realization $(E,A,B,C,\tau)$ as in \eqref{eq:SSDescriptordelay} is looked for.


\section{Delay Loewner framework}\label{sec:delayinterpolation}

Before introducing the main result of this section, let us consider the following representation of system \eqref{eq:SSDescriptordelay}, which will be useful along the rest of the paper. 

\begin{lemma}\label{lemma:DSDrep} Given $\mathbf{H_d} =(E,A,B,C,\tau)$, its transfer function $H_d(s)$ can be decomposed as:
\begin{equation}
H_d(s) = G\big(f(s)\big)e^{s\tau}
\end{equation}
where $G(s)$ is the transfer function of the delay-free model $\textbf G=(E,A,B,C)$ as in \eqref{eq:SSDescriptor} and $f(s) = se^{s\tau}$. 
\end{lemma}
\begin{proof} The result is straightforwardly obtained by injecting $f(s)$ in \eqref{eq:SSDescriptor} as:
\begin{eq}
\begin{array}{rcl}
H_d(s) &=& C(sE-Ae^{-s\tau})^{-1}B \\
& =& C(se^{s\tau}E-A)^{-1}Be^{{s\tau}} \\
&= &G(se^{s\tau})e^{s\tau}.
\end{array}
\end{eq}
\end{proof}

Then one extension of Theorem~\ref{thm:Loewner}  which makes feasible the interpolation with a single delay descriptor system as defined in \eqref{eq:SSDescriptordelay} can be done by using $f(s)$ as a variable substitution and applying the standard Loewner framework to the new transformed data. This first main result can be stated as follows:
\begin{theorem}[Delay Loewner framework]\label{thm:delayLowner} 
Let us consider $\rho = \nu = r$, $\tau \in \mathbb{R}$ and given  $(\lambda_i,\mathbf{r}_i,\mathbf{w}_i) $ and $(\mu_j,\mathbf{l}_j,\mathbf{v}_j)$  the \emph{right and left interpolation data} respectively, as stated in \eqref{eq:rightinterdata}-\eqref{eq:leftinterdata}. Assuming that $f(s)= se^{s\tau}$ is one-to-one in the interpolation points domain\footnote{This means that for any $h_1, h_2 \in \{\lambda_1,\dots,\lambda_r\}\cup\{\mu_1,\dots,\mu_r \} $, then $f(h_1) \neq f(h_2)$ if $h_1\neq h_2$, where $f(s) = se^{s\tau}$.} and let $\mathbf{G} = (\hat{E},\hat{A},\hat{B},\hat{C})$ be a realization satisfying right and left constraints from the data $(f(\lambda_i),\mathbf{r}_i,\mathbf{w}_ie^{-\lambda_i\tau}) $ and $(f(\mu_j),\mathbf{l}_j,\mathbf{v}_je^{-\mu_i\tau})$ constructed with Theorem \ref{thm:Loewner}. Then $\mathbf{H_d} = (\hat{E},\hat{A},\hat{B},\hat{C},\tau)$ satisfies the \emph{right}:
\begin{eq}
H_d(\lambda_i)\mathbf{r}_i = \mathbf{w}_i,~ i = 1,\dots r
\label{eq:delayRightConst}
\end{eq}
and \emph{left constraints}:
\begin{eq}
\mathbf{l}_jH_d(\mu_j) = \mathbf{v}_j,~ j = 1,\dots r
\label{eq:delayLeftConst}
\end{eq}
for the given right and left interpolation data.
\end{theorem}
\begin{proof}
The result for the right constraints \eqref{eq:delayRightConst} is obtained as follows: first note that if the delay-free model $G(s)$ satisfies the right constraints for $(f(\lambda_i),\mathbf{r}_i,\mathbf{w}_ie^{-\lambda_i\tau})$, then one obtains:
\begin{eq}
G(f(\lambda_i))\mathbf{r}_i  = \mathbf{w}_ie^{-\lambda_i\tau} 
\end{eq}
then, it equivalently follows that:
\begin{eq}
G(f(\lambda_i))e^{\lambda_i\tau}\mathbf{r}_i =  \mathbf{w}_i  
\end{eq}
and by invoking Lemma \ref{lemma:DSDrep}, we obtain the result:
\begin{eq}
H_d(\lambda_i)\mathbf{r}_i =  \mathbf{w}_i.
\end{eq}
The left data constraints \eqref{eq:delayLeftConst} is similarly obtained.
\end{proof}

Theorem \ref{thm:delayLowner} provides a method to construct a model $\mathbf{H_d} = (E,A,B,C,\tau)$ whose transfer function $H_d(s) = C(sE-Ae^{-s\tau})^{-1}B$ interpolates given right and left constraints. This is possible by noticing that the problem can be rewritten as right and left interpolation constraints for the delay-free for which a realization is obtained by the standard Loewner framework as in Theorem \ref{thm:Loewner}. A similar reasoning enables the generalization of Theorem~\ref{thm:DerivativeLoewner} as stated follows.  

\begin{theorem}[Derivative delay Loewner framework]\label{thm:DerivdelayLoewner} 
Let us consider a given system represented by its transfer function $H(s)$,  $r$ shift points $\{s_1,\dots,s_r\} \in \mathbb{C}$ and $r$ tangential directions  $\{\mathbf{l}_1,\dots, \mathbf{l}_r\} \in \mathbb{C}^{1\times n_y}$, $\{\mathbf{r}_1,\dots,\mathbf{r}_r\}\in \mathbb{C}^{n_u \times 1}$. We assume that for all $k \neq m $, $f(s_k) \neq f(s_m)$, where $f(s) = se^{s\tau}$ ($f$ is one-to-one in the interpolation points domain). The $r$-dimensional  single delay model $\mathbf{\hat{H}} = (\hat{E},\hat{A},\hat{B},\hat{C},\tau)$, as in \eqref{eq:SSDescriptordelay}, interpolates $H(s)$ as follows, for $k=1,\dots , r$:
\begin{equation}
 H({s}_k)\mathbf{r}_k = \hat{H}({s}_k)\mathbf{r}_k,\,\, \mathbf{l}_kH(s_k) = \mathbf{l}_k\hat{H}({s}_k)\,
\end{equation}
\begin{equation}
\mathbf{l}_kH'(s_k)b_k = \mathbf{l}_kH'(s_k)\mathbf{r}_k,
\label{pbh21}
\end{equation} 
if only if the $r$-dimensional delay-free model $\mathbf{G} =  (\hat{E},\hat{A},\hat{B},\hat{C})$  is constructed with the derivative Loewner framework as in Theorem \ref{thm:DerivativeLoewner} for the shift points:
\begin{equation}
\big(\sigma_1,\dots,\sigma_r\big)  = \big(f(s_1),\dots,f(s_r)\big),\end{equation} 
and the transfer function evaluation:
\begin{equation}\big( G(\sigma_1), \dots, G(\sigma_r)\big) =\big( H(s_1) e^{-s_1\tau}, \dots,  H(s_r) e^{-s_r\tau}\big) 
\end{equation} 
and the derivative transfer function evaluation:
\begin{equation} \big(G'(\sigma_1),\dots, G'(\sigma_r)\big) = \big(F_1, \dots , F_r\big) 
\end{equation} 
where for $i = 1,\dots r$:
\begin{eq}
\begin{array}{rcl}
G'(\sigma_i) &=& F(H(s_i),H'(s_i),s_i) = F_i \\
&=& \big(H'(s_i) - \tau H(s_i)\big)\bigg(\frac{e^{-2 s_i\tau}}{1+\tau s_i}\bigg)
\end{array}
\end{eq}
\end{theorem}
\begin{proof}
First, one can note that a single delay descriptor system can be expressed as
\begin{equation}\label{descriptortheorem}
\hat{ H }_d(s) = \hat{C}(se^{s\tau}\hat{E}-\hat{A})^{-1}\hat{B}e^{{s\tau}} = G(f(s))e^{s\tau}  
\end{equation}
where
$G(s)$ is a descriptor system whose representation is $(\hat{E},\hat{A},\hat{B},\hat{C})$ and $f(s) = se^{s\tau}$. 
Thus one can use Loewner matrices to construct the realization of system $G(s)$ for the shift points $(\sigma_1,\dots \sigma_r) = (f(s_1),\dots,f(s_r))$  whose transfer function data are $(G(\sigma_1),\dots,G(\sigma_r)) = (H(s_1)e^{-s_1\tau},\dots,H(s_r)e^{-s_r\tau}) $. For the transfer function derivative data, one can take the derivative of  \eqref{descriptortheorem} written as $G(f(s)) = H_d(s)e^{-s\tau}$ with respect to $s$ as follows 
\[
G(f(s))f'(s) = \hat{H}_d'(s)e^{-s\tau}-\tau \hat{H}_d(s)e^{-s\tau}  
\]
and by solving the equation for $G'(s_k)$ one can obtain the result.
\end{proof}

This theorem allows to obtain a single delay descriptor system which interpolates any given transfer function $H(s)$. This can also be used in the case of data obtained through a signal generator, considering that the derivative is accessible as well. Applications of this result can be found in section \ref{sec:applications}. 

Now that the delay interpolation framework has been established, one might be interested in obtaining a good interpolant in the sense of the $\mathcal{H}_2$-norm as formulated in Problem \ref{pb:General}. We will now formulate mathematical conditions to select the optimal shift complex points $s_i$ and tangential directions $\mathbf{r}_i$ and $\mathbf{l}_i$. 

\section{$\mathcal{H}_2$ model reduction optimality conditions}\label{sec:Optimality}

\subsection{Preliminary results in $\mathcal{H}_2$ model reduction optimality conditions}

The first-order optimality conditions for the delay-free $\mathbf{\hat{H}} = (\hat{E},\hat{A},\hat{B},\hat{C},0)$ Problem \ref{pb:General} in terms of poles and residues are given in Theorem \ref{thm:opt}.
\begin{theorem}[\cite{gugercin2008h_2}]\label{thm:opt}
Assume that $\mathbf{H}$ and $\mathbf{\hat{H}}$ have semi-simple poles and suppose that $\mathbf{\hat{H}}$ is a $r^{th}$-order finite-dimensional model with transfer function:
\begin{equation}\label{sysr}
\hat{H}(s) = \sum_{k =1}^r \frac{\hat{c}_k\hat{b}_k^T}{s-\hat{\lambda}_k}.
\end{equation}
If $\mathbf{H}, \mathbf{\hat{H}}\in \mathcal{H}_2$ and $\mathbf{\hat{H}}$ of the form \eqref{eq:SSDescriptor}, is a local minimum of the $\mathcal{H}_2$ delay-free approximation problem, then the following interpolation equations hold:
\begin{equation}
H(-\hat{\lambda}_k)\hat{b}_k = \hat{H}(-\hat{\lambda}_k)\hat{b}_k,\,\,\hat{c}^T_kH(-\hat{\lambda}_k) = \hat{c}^T_k\hat{H}(-\hat{\lambda}_k)\,
\label{eq:pbh2}
\end{equation}
\begin{equation}
\hat{c}^T_kH'(-\hat{\lambda}_k)\hat{b}_k = \hat{c}^T_kH'(-\hat{\lambda}_k)\hat{b}_k,
\label{eq:pbh21}
\end{equation}
for all $k= 1,\dots,r$ where $\hat{\lambda}_k$ are the poles of $\mathbf{\hat{H}}$ and $\hat{b}_k$ and $\hat{c}_k$ are its tangential directions, respectively.
\end{theorem}
As previously for the interpolation conditions, the $\mathcal{H}_2$-optimality conditions are now extended to the single delay case. 

\subsection{Results in single delay model reduction $\mathcal{H}_2$  optimality conditions}

\begin{proposition}  Using the notation in Lemma \ref{lemma:DSDrep}, $\lambda \in \mathbb{C}$ is a pole of $H_d(s) = C(sE-Ae^{-s\tau})^{-1}B$ if and only if  $f(\lambda) \in \mathbb{C}$ is a generalized eigenvalue of the pair $(A,E)$. 
\end{proposition}  
\begin{proof}
$\lambda$ is pole of $H(s)$ $\iff $ $w(\lambda)$ is pole of $G(s)$ $\iff $  $w(\lambda) \in$ $\sigma(A,E)$, the spectrum of  pencil $(A,E)$. 
\end{proof}

Now let recall that the Lambert $W$ function is a multivalued (except at $0$) function associating for each $k^{th}$ complex branch, a complex number $W_k(z)$ such that
\[  z = W_k(z)e^{W_k(z)}, \hspace{0.5cm} k\in \mathbb{Z}.\]
Consequently any TDS can be viewed as a system with an infinite number of poles. 

Now, as in the delay-free case, analog optimality conditions are derived in the case where the approximation model has the form \eqref{eq:SSDescriptordelay}.

\begin{theorem}[Delay $\mathcal{H}_2$-optimality conditions]\label{thm:DelayOptim}
Assume that $\mathbf{H}$ and $\mathbf{\hat{H}_d}$ have semi-simple poles and suppose that $\mathbf{\hat{H}_d}$ is a $r^{th}$-order single delay model whose transfer reads:
\begin{equation}\label{sysr}
\hat{H}_d(s) = \hat{C}(s\hat{E} - \hat{A}e^{-s\tau})^{-1}\hat{B} =  \sum\limits_{k =1}^r \frac{\hat{c}_k\hat{b}_k^T}{s-\hat{\alpha}_ke^{-s\tau}}.
\end{equation}
If $\mathbf{H}, \mathbf{\hat{H}_d}\in \mathcal{H}_2$ and $\mathbf{\hat{H}}$ is a local minimum of the $\mathcal{H}_2$ approximation problem, then the following interpolation equations hold:
\begin{equation}\label{eq:pbh2delay}
\hspace{-0.03cm} H(-\hat{\lambda}_{k,p})\hat{b}_k = \hat{H}_d(-\hat{\lambda}_{k,p})\hat{b}_k,~\hat{c}^T_kH(-\hat{\lambda}_{k,p}) = \hat{c}^T_k\hat{H}_d(-\hat{\lambda}_{k,p})\,
\end{equation}
\begin{equation}\label{eq:pbh21delay}
\hat{c}^T_kH(-\hat{\lambda}_{k,p})\hat{b}_k = \hat{c}^T_k\hat{H}_d(-\hat{\lambda}_{k,p})\hat{b}_k,
\end{equation}
for all $p= 1,\dots,r$ and $k \in \mathbb{Z}$ where, for given $p$, $\hat{\lambda}_{k,p}$ are defined by:
\begin{equation}\label{eq:polesinfinity}
\lambda_{k,p}= \frac{1}{\tau}W_k(\tau \alpha_p) 
\end{equation} where $W_k$ is the $k^{th}$ branch of the multivalued Lambert function.
\end{theorem}  

\begin{proof} The proof is similar those of Theorem  \ref{thm:opt} as in \cite{gugercin2008h_2} using the infinite poles and residues decomposition of model $\hat{H}_d(s)$.
\end{proof}

Theorem \ref{thm:DelayOptim} states that the optimal model $\mathbf{\hat{H}_d}$ of Problem \ref{pb:General}, if it exists, satisfies an infinite number of optimality conditions related with the Lambert $W$ function and the general eigenvalues of $(\hat{E},\hat{A})$. Nevertheless, given $\tau \in \mathbb{R}$, as $\hat{H}_d(s) = \hat{C}(s\hat{E}-\hat{A}e^{-s\tau})^{-1}\hat{B}$ is parametrized by a finite number of variables, and it can be shown that it lives in a sub-manifold of dimension $n(n_u+n_y)$. This can be simply shown by noticing that there is a simple isomorphism between $(E,A,B,C,\tau)$ and $(E,A,B,C,0)$ and the  last one is parametrized by $n(n_u+n_y)$ variables as it can be seen in \cite{byrnes1979applications,dooren2007}. All the optimality conditions cannot be achieved in the general case. However, as stated in the following proposition, in a given particular case, the infinite optimality conditions fall into an equivalent finite number of relationships.
\begin{proposition} Assuming the same hypotheses of Theorem \ref{thm:DelayOptim} about $\mathbf{H}$ and $\mathbf{\hat{H}}$. Moreover if model $\mathbf{H} = (E,A,B,C,\tau)$, then the infinite optimality conditions of Theorem \ref{thm:DelayOptim} can be resumed to a finite number of optimality conditions as follows for $p=1,\dots r$ :
\begin{equation}\label{eq:pbh2delayfinite}
H(-\hat{\lambda}_{1,p})\hat{b}_k = \hat{H}(-\hat{\lambda}_{1,p})\hat{b}_k,\,\,\hat{c}^T_kH(-\hat{\lambda}_{1,p}) = \hat{c}^T_k\hat{H}(-\hat{\lambda}_{1,p})\, 
\end{equation}
\begin{equation}\label{eq:pbh21delayfinite}
\hat{c}^T_kH(-\hat{\lambda}_{1,p})\hat{b}_k = \hat{c}^T_k\hat{H}(-\hat{\lambda}_{1,p})\hat{b}_k,
\end{equation}
where $\lambda_{1,p}= \frac{1}{\tau}W_1(\tau \alpha_p) $ and $W_1$ is the evaluation of the  Lambert function along its $1^{st}$ branch.
\end{proposition} 
\begin{proof}
One have to prove that the finite conditions \eqref{eq:pbh2delayfinite}-\eqref{eq:pbh21delayfinite} imply \eqref{eq:pbh2delay}-\eqref{eq:pbh21delay}. This is possible due to the fact that:
\[f(\lambda_{1,p}) = \lambda_{1,p} e^{\lambda_{1,p}\tau} = \lambda_{k,p} e^{\lambda_{k,p}\tau} = f(\lambda_{k,p}),~\forall k\in \mathbb{Z}\]
Thus, using the decomposition given in Lemma \ref{lemma:DSDrep}, it can be shown that:
\begin{equation*}
\begin{array}{rcl}
H(\lambda_{k,p}) &=& G(f(\lambda_{k,p}))e^{\tau\lambda_{k,p}}  \\ &=& G(f(\lambda_{1,p}))e^{\tau\lambda_{1,p}}e^{\tau(\lambda_{k,p} - \lambda_{1,p})} \\
 &=&  H(\lambda_{1,p}) e^{\tau(\lambda_{k,p} - \lambda_{1,p})} 
\end{array}
\end{equation*}
and finally
\[ H(\lambda_{k,p})\mathbf{b}_k = CH(\lambda_{1,p})\mathbf{b}_k = C\hat{H}(\lambda_{1,p})\mathbf{b}_k  = \hat{H}(\lambda_{k,p})\mathbf{b}_k ,\] $\forall k \in \mathbb{Z}$, where $C = e^{\tau(\lambda_{k,p}-\lambda_{1,p})}$. The reasoning is analog for the right and derivative constraints, which concludes the proof.
\end{proof}

Now that the optimality conditions have been derived, next section is dedicated to the derivation of an algorithm, based on \textbf{TF-IRKA} \cite{beattie2012realization}, denoted \textbf{dTF-IRKA} for \emph{delay Transfer Function Iterative Rational Krylov Algorithm}, which allows to obtain a sub-optimal model of the form \eqref{eq:SSDescriptordelay}, satisfying  $n(n_u+n_y)$ optimality conditions.  


\section{Delay TF-IRKA algorithm}\label{sec:delayTFIRKA}

The algorithm proposed in this section permits to derive a system which satisfies the optimality conditions for $r$ complex points. The idea behind is based on \textbf{TF-IRKA} \cite{beattie2012realization} which finds a model satisfying the optimality conditions in \eqref{eq:pbh2}-\eqref{pbh21} using a fixed point iteration.  For each iteration the new shift points will be the poles located in the $1^{st}$ branch of the Lambert function, only.  This algorithm is celebrated as \textbf{delay TF-IRKA} (or \textbf{dTF-IRKA}) and is summed up as follows:

\begin{algorithm}[H]
\caption{ dTF-IRKA }
\label{TFIRKAalgo}
\begin{algorithmic}[1]
\State \textbf{Initialization:} transfer function $H(s)$, approximation order $r\in \mathbb{N}^*$, $\sigma^0 = \{ {\sigma_1^0,\dots,\sigma_r^0}\}\in \mathbb{C}$ initial interpolation points and tangential directions $\{b_1,\dots,b_r\} \in \mathbb{C}^{n_u \times 1}$ and $\{c_1,\dots,c_r\}\in \mathbb{C}^{n_y\times 1}$, $W_1$ the first branch of the Lambert function.
\While{not convergence}
	\State\textbf{Build}  $\big(\hat{E}$, $\hat{A}$, $\hat{B}$, $\hat{C},\tau\big)$ using Theorem~\ref{thm:DerivdelayLoewner}.
	\State \textbf{Solve} the generalized eigenvalue problem in $x_i^{(k)}$, $y_i^{(k)}$ and $\lambda_i^{(k)}$, for $i=1,\dots,r$
	\begin{eq}
	\begin{array}{rcl}
	\hat{A}^{(k)}x_i^{(k)}&=&\lambda_i^{(k)}\hat{E}^{(k)}x_i^{(k)} \\
	y_i^{(k)*}\hat{E}^{(k)}x_j^{(k)} &=& \delta_{i,j}
	\end{array}
	\end{eq}
	\State \emph{Set}, for $i=1,\dots,r$
	\begin{eq}
	\begin{array}{rcl}
	\sigma_i^{(k+1)} &\gets&  -\frac{1}{\tau}W_1\big(\tau\lambda_i^{(k)}\big) \\
	b_i^{(k+1)T} &\gets& y_i^{(k)}\hat{B}^{(k)}\\ 
	c^{(k+1)}_i &\gets& \hat{C}^{(k)}x_i^{(k)}
	\end{array}
	\end{eq}
\EndWhile
	\State \textbf{Ensure} conditions \eqref{eq:pbh2delayfinite}-\eqref{eq:pbh21delayfinite} are satisfied.
	\State\textbf{Build}  $\mathbf{\hat{H}} = \big(\hat{E}$, $\hat{A}$, $\hat{B}$, $\hat{C},\tau\big)$.

\end{algorithmic}
\end{algorithm}

If the algorithm converges, the approximation model will satisfy optimality conditions given by \eqref{eq:pbh2delayfinite}-\eqref{eq:pbh21delayfinite} and will therefore be suboptimal. The  \textbf{dTF-IRKA} then allows to obtain good (in the sense of the metric given in Problem \ref{pb:General}) shift points and tangential directions for which the interpolation problem will lead to a good approximation model. 

As a remark, one should note that the Lambert function evaluated in the $1^{st}$ branch can sometimes associate a real number to a complex one. In this way, the shift points might not be a closed set (by conjugation) and the obtained single delay interpolation model will not have a real representation. To avoid this, one should enforce at each iteration the shift points to be closed by conjugation. 
 
\section{Applications}\label{sec:applications}

This section is dedicated to the application of both methods proposed in Sections \ref{sec:delayinterpolation} and  \ref{sec:delayTFIRKA}, namely, the delay model interpolation and optimal $\mathcal H_2$ model approximation. We will emphasize the potential benefit and effectiveness of the proposed approach.

\subsection{Example 1: rational interpolation}

Let us consider a dynamical model governed by the following delay model $\mathbf{H} \in \mathcal{H}_2$ whose transfer function is given as
\begin{equation}\label{eq:example1TF}
H(s) =\dfrac{2s+1.3 e^{-s}}{s^2+1.3s e^{-s}+0.3e^{-2s}}.
\end{equation}

First, model \eqref{eq:example1TF} (which is obviously of order 2) is approximated by a delay-free model of order $r=2$ using the \textbf{TF-IRKA} (Figure \ref{fig:Exemple1a}, green dashed dotted curve). It is also interpolated using the delay Loewner framework with derivatives as stated in Theorem \ref{thm:DerivdelayLoewner} whose delay is set to $\tau =1$, at the shift points $s_1 = 0.1$ and $s_2 = 1$ (Figure \ref{fig:Exemple1a}, red dashed thick curve). All the results are reported on  Figure \ref{fig:Exemple1a}, and compared to the original model $H(s)$ (solid blue line). 

\begin{figure}[here]
  	\centering
	\includegraphics[width=0.75\textwidth]{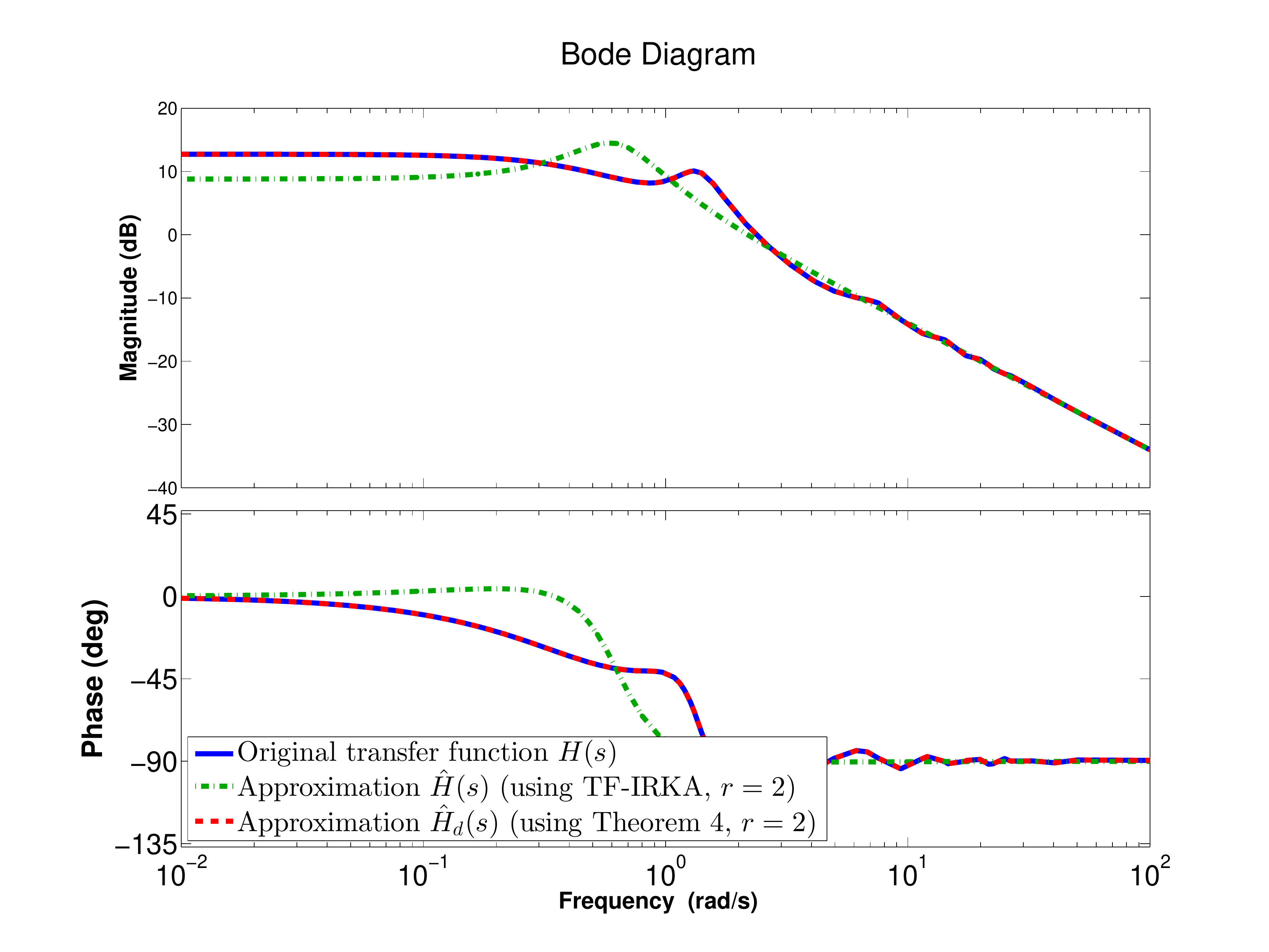}
	 \caption{Bode diagram of original model (blue solid line), model of order $r=2$ approximated with \textbf{TF-IRKA} (green dashed dotted curve) and delay interpolation model using Theorem \ref{thm:DerivdelayLoewner} of order $r=2$ (red dashed line).}   
	 \label{fig:Exemple1a}
\end{figure}

Figure \ref{fig:Exemple1a} shows that model defined in \eqref{eq:example1TF} is well interpolated by a delay model obtained by Theorem \ref{thm:DerivdelayLoewner}, for any interpolation points. Indeed, since the transfer function \eqref{eq:example1TF} has a realization of the form \eqref{eq:SSDescriptordelay} of order 2 it can be reconstructed using the Theorem \ref{thm:DerivdelayLoewner}. Figure \ref{fig:Exemple1b} shows quite similar results but where \textbf{TF-IRKA} has targeted an order $r=4$. 
\begin{figure}[here]
  	\centering
	\includegraphics[width=0.75\textwidth]{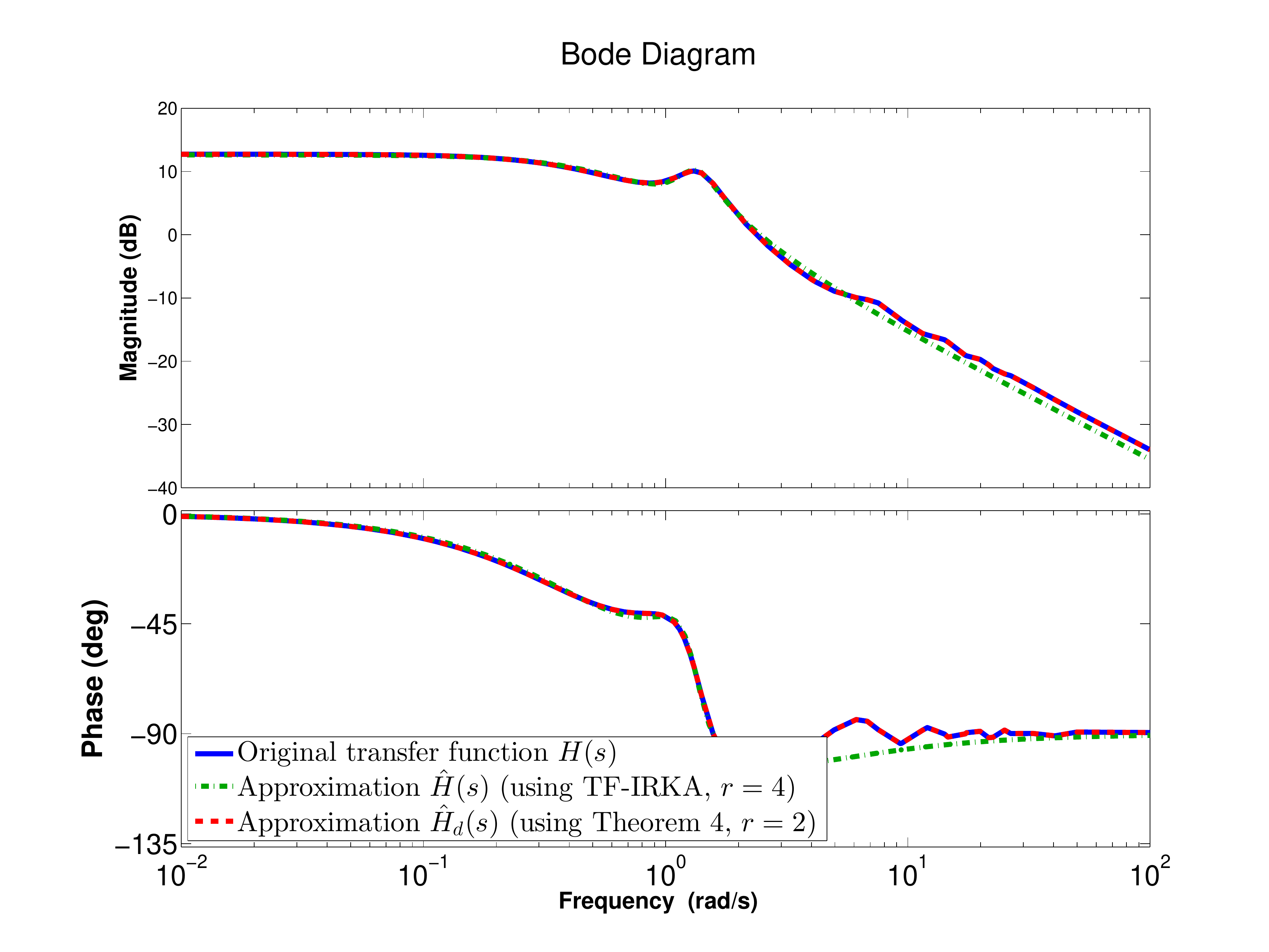}
	 \caption{Bode diagram of original model (blue solid line), model of order $r=4$ approximated with \textbf{TF-IRKA} (green dashed dotted curve) and delay interpolation model using Theorem \ref{thm:DerivdelayLoewner} of order $r=2$ (red dashed line).}   
	 \label{fig:Exemple1b}
\end{figure}

\begin{figure}[here]
  	\centering
	\includegraphics[width=0.75\textwidth]{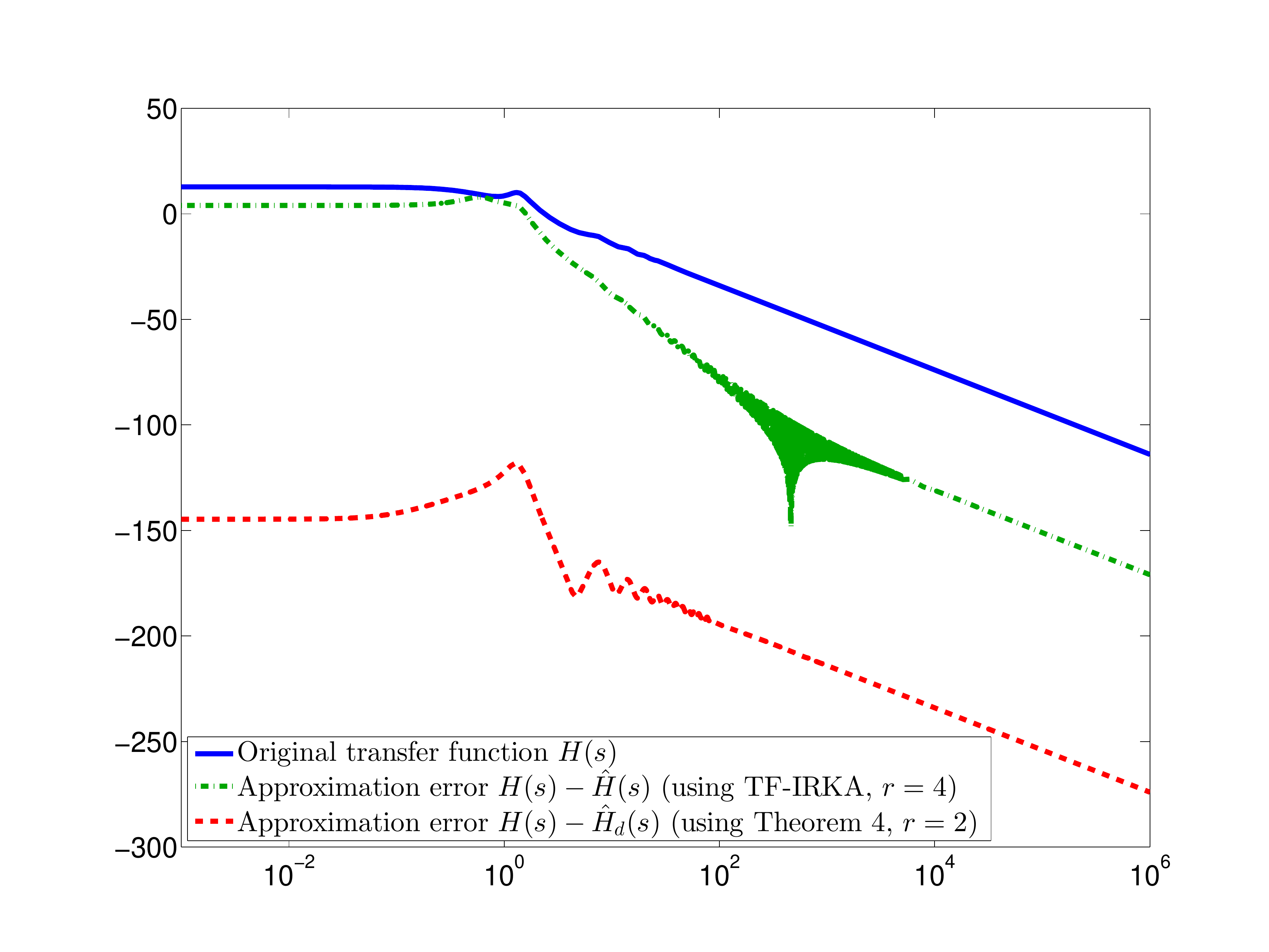}
	 \caption{Singular value frequency response diagram of original model (blue solid line), approximation error with model of order $r=2$ obtained with \textbf{TF-IRKA} (green dashed dotted curve) and approximation error of the delay interpolation model using Theorem \ref{thm:DerivdelayLoewner} of order $r=2$ (red dashed line).}   
	 \label{fig:Exemple1b_error}
\end{figure}

This specific example clearly emphasizes the fact that, if the original model is a delay model, the counter part of obtaining a good delay-free approximation  (\emph{e.g.}, using \textbf{TF-IRKA}) is to increase the approximation order (here the original model of order 2 must be approximated with an order 4 to well recover the frequency behaviour). As illustrated on Figure  \ref{fig:Exemple1b_error}, even with an order $r=4$, the delay-free model cannot perfectly recover the original infinite dimensional model, while the delay model (obtained by Theorem \ref{thm:DerivdelayLoewner}) provides perfect matching (subject to numerical machine precision errors). On the other hand, the proposed delay Loewner framework allows to find an exact realization.

\subsection{Example 2: optimal approximation and method scalability} 

Let us now consider the SISO Los-Angeles Hospital model extracted from the COMP$l_eib$ library \cite{COMPleib2003} whose order is $n=48$, denoted $\mathbf{H}_{build} = C(sI_{48}-A)^{-1}B \in \mathcal{H}_2$. In order to fit the framework proposed in this paper, a delay model is constructed by injecting an internal delay $\tau = 0.01$ to all states, \emph{i.e.}, $\mathbf{H}_{delay} = C(sI_{48}-Ae^{-s\tau})^{-1}B $. This last transfer function is firstly interpolated on the basis of realisation of order $r=10$ by applying the delay Loewner framework from Theorem \ref{thm:DerivdelayLoewner} using $10$ real shift points logarithmically spaced from $0.1$ to $1$. Then, an approximation is obtained using the  \textbf{dTF-IRKA} algorithm proposed in Section \ref{sec:delayTFIRKA}. Figure \ref{fig:Exemple2} compares the Bode plot of these models. 
\begin{figure}[H]
  	\centering
	\includegraphics[width=0.75\textwidth]{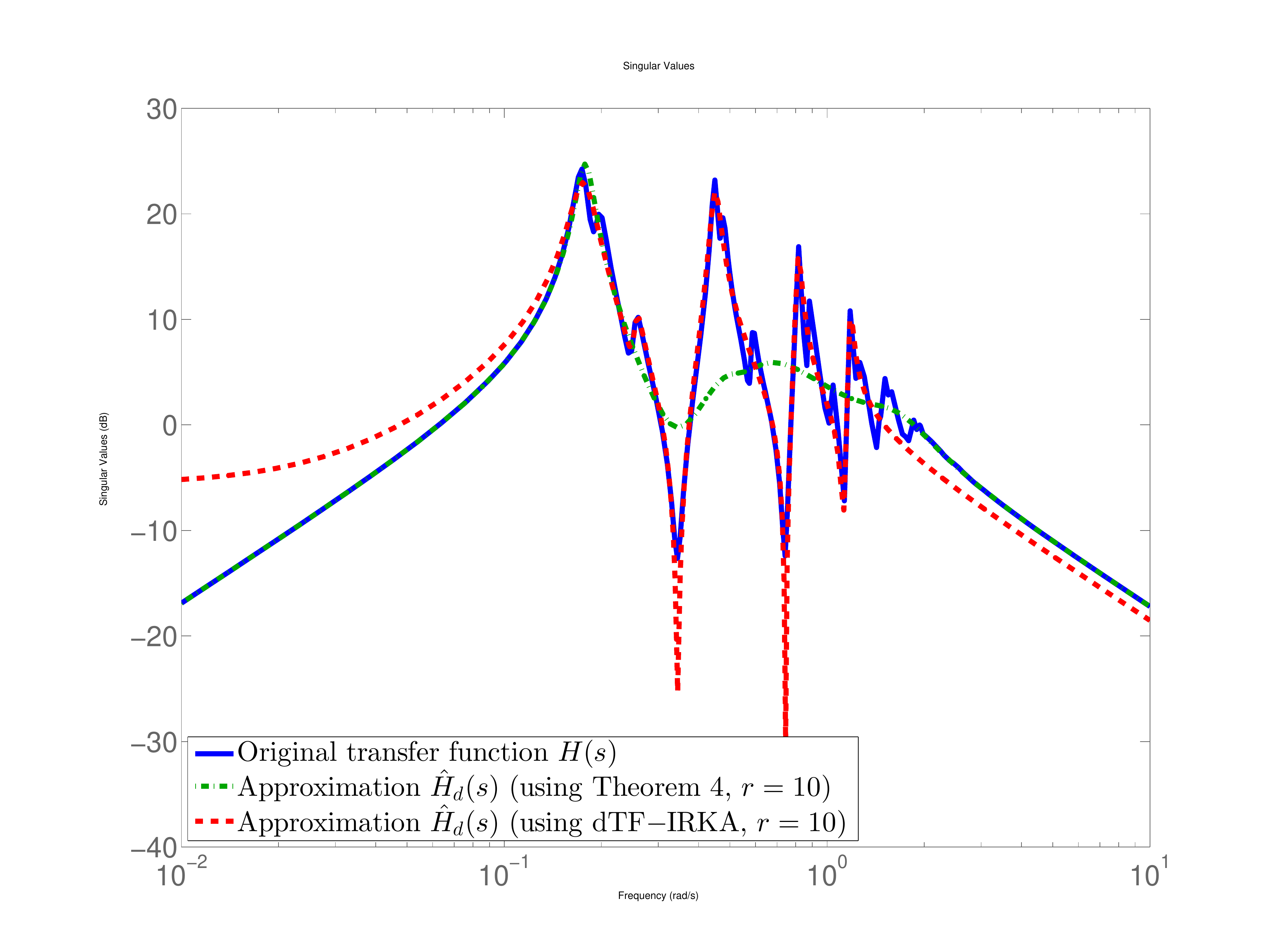}
	 \caption{Bode diagram of original model (blue curve),  delay Loewner interpolation model of order $r=10$ (green dashed curve) and \textbf{dTF-IRKA} of order $r=10$ (red dashed curve).}
	 \label{fig:Exemple2}
\end{figure}

As clearly shown on Figure \ref{fig:Exemple2}, the proposed \textbf{dTF-IRKA} allows to obtain shift points and tangential directions for which the interpolated delay model is much more accurate than the approximation using random shift points. This shows the scalability of the proposed approach for larger models.

\section{Conclusions and perspectives}

In this paper, the problem of interpolating and approximating any dynamical model (provided its transfer function or its evaluation at given points) by a single time-delay finite dimensional one is analysed. Firstly, we present an extended framework which generalizes the Loewner one \cite{mayo2007framework} to the case where the interpolant is a single time-delay model. Then, as a second contribution, the $\mathcal{H}_2$-optimality conditions are derived to solve Problem \ref{pb:General}, leading to an infinite set of conditions. Finally, an algorithm, denoted \textbf{dTF-IRKA}, allowing to obtain a model which satisfies a finite number of optimality conditions is developed and successfully applied to some numerical examples\footnote{Reader should note that the \textbf{dTF-IRKA} will be made available in the MORE toolbox \cite{vuillemin2012}, developed by the Onera research group.}.

One weakness of the proposed method, is the fact that one should know in advance the delay value $\tau$. Future works will investigate this issue by taking into consideration the delay as an optimization variable in the $\mathcal{H}_2$ optimization problem. The extension to the multiple delay case will also be addressed in future works.

\bibliography{igorBiblio}

\end{document}